\documentclass[12pt, reco]{article}
\usepackage{amssymb,amscd,amsmath,amsthm}
\usepackage[colorinlistoftodos]{todonotes}
\usepackage[colorlinks=true, allcolors=blue]{hyperref}
\usepackage{lmodern}
\usepackage[english]{babel}
\usepackage{color}
\usepackage{graphicx}
\usepackage{tikz}
\usetikzlibrary{positioning, calc, arrows.meta}
\usepackage[margin=1in]{geometry}
\newtheorem{theorem}{Theorem}[section]

\newtheorem{lemma}[theorem]{Lemma}
\newtheorem{definition}[theorem]{Definition}

\newtheorem{remark}[theorem]{Remark}
\def\Tr{\mathrm{Tr}}
\def\id{{\bf 1}\!\!{\rm I}}

\title{Dynamics of Matrix Product States in the Heisenberg Picture: Projectivity, Ergodicity, and Mixing}

\begin{document}


\maketitle

\centerline{ \author{  Abdessatar Souissi}}

\centerline{  Department of Management Information Systems, College of Business and Economics, }
\centerline{Qassim University, Buraydah 51452, Saudi Arabia}

\centerline{\textit{a.souaissi@qu.edu.sa}}
\vskip0.5cm

\centerline{\author{ Amenallah Andolsi}}

\centerline{  Nuclear Physics and High Energy Research Unit, Faculty of Sciences of Tunis,}
\centerline{ Tunis El Manar University, Tunis, Tunisia }

\centerline{\textit{amenallah.andolsi@fst.utm.tn }}

\tableofcontents

 \begin{abstract}
This paper introduces a Heisenberg picture approach to Matrix Product States (MPS), offering a rigorous yet intuitive framework to explore their structure and classification. MPS efficiently represent ground states of quantum many-body systems, with infinite MPS (iMPS) capturing long-range correlations and thermodynamic behavior. We classify MPS into projective and non-projective types, distinguishing those with finite correlation structures from those requiring ergodic quantum channels to define a meaningful limit. Using the Markov-Dobrushin inequality, we establish conditions for infinite-volume states and introduce ergodic and mixing MPS. As an application, we analyze the depolarizing MPS, highlighting its lack of finite correlations and the need for an alternative ergodic description. This work deepens the mathematical foundations of MPS and iMPS, providing new insights into entanglement, phase transitions, and quantum dynamics.
\end{abstract}


\textbf{Subjectclass:}{ 46L55, 81P45, 81Q10, 82B10, 81R15}

\section{Introduction}
Understanding the dynamics of highly entangled and strongly correlated quantum systems remains one of the most challenging problems in quantum theory. In quantum information, quantum channels offer a robust framework for modeling the evolution of quantum systems, providing critical insights into the dynamics of quantum processes and the transmission of information \cite{ SN24, SB24, Mov22}.

On the other hand, algebraic states—defined as positive, linear, and unital functionals on operator algebras—serve as a  powerful tool for characterizing both pure and mixed quantum states. These states play a central role in defining equilibrium through Kubo-Martin-Schwinger (KMS) states \cite{PS75,   Go83}, which generalize the concept of thermal equilibrium to infinite-dimensional systems. Quantum Markov chains (QMCs) \cite{  AccFri83, ASS20} further extend this framework by providing an effective way to describe the dynamics of quantum systems exhibiting Markovian behavior.

Hidden quantum Markov models (HQMMs) \cite{ AGLS24Q, SS23} offer a refined approach for analyzing stochastic processes with memory effects, governed by underlying QMCs. These models have shown particular promise in the emerging field of quantum machine learning, where they enable novel applications in areas such as pattern recognition, probabilistic inference, and optimization \cite{GSPEC19, MVVMSH18, Li24}.

Finitely correlated states \cite{fannes, fannes2, fannes3} provide an effective approach for modeling systems with finite-range correlations, significantly advancing  quantum  lattice models. It was demonstrated in \cite{fannes} that these states are equivalent to generalized valence bond states. Among these, the states exhibiting pure exponential clustering are particularly notable, as they arise as unique ground states of translation-invariant, finite-range Hamiltonians, which are characterized by a non-zero spectral gap.

Matrix Product States (MPS) \cite{perez2007matrix, verstraete2008matrix, orus2014practical, P19, G06} provide an effective framework for representing wave functions of one-dimensional quantum many-body systems. Initially introduced in the context of Density-Matrix Renormalization Group (DMRG) methods \cite{white1992density, schollwock2011density, Mcc07, P19, T10, B21}, the MPS formalism has since been widely extended and has proven highly effective in a range of machine learning applications \cite{Dh22, Meng23, LiZh23, Han18}. A Matrix Product State (MPS) with periodic boundary conditions is represented by a sequence of matrices \(\{A_i\}\), where each matrix corresponds to a local basis state \(|i\rangle\) of a Hilbert space \(\mathcal{H}\). MPS formalism provides a compact and efficient representation of ground states in one-dimensional quantum many-body systems. A finite MPS is given by
\begin{equation*}
|\psi_N\rangle = \sum_{i_1, i_2, \dots, i_n} \mathrm{Tr}\left( A_{i_1}^{[1]} A_{i_2}^{[2]} \cdots A_{i_n}^{[n]} \right) |i_1 i_2 \dots i_n\rangle \in \mathcal{H}^{\otimes n}
\end{equation*}
where the matrices $A_{i}^{[j]}$ encode the entanglement structure of the state. This framework naturally extends to infinite MPS (iMPS) \cite{B14, CS2010, C10}, which describe quantum states on infinite lattices and allow for the study of long-range correlations and thermodynamic limit behavior. A key feature of iMPS is their ability to capture asymptotic properties without being restricted by finite-size effects.
To rigorously analyze correlation structures in MPS, we define the sequence of algebraic states
\begin{equation*}
\varphi_n(X) = \frac{\big\langle \psi_{n+1} \big| X\otimes \id \big| \psi_{n+1}\big\rangle  }{\langle \Psi_{n+1} |\Psi_{n+1}\rangle}, \quad X\in\mathcal{B}_{[1,n]} = \bigotimes_{[1,n]} \mathcal{B}(\mathcal{H}).
\end{equation*}
This construction enables a systematic study of transition amplitudes and correlation decay properties. As $n \to \infty$, this sequence defines a well-posed thermodynamic limit, characterizing the emergence of ergodicity and mixing properties \cite{bau2013}. In particular, whether an iMPS exhibits finite or infinite correlation length depends on the spectral properties of the corresponding transfer operator.

This paper introduces a Heisenberg picture framework for MPS  by representing the sequence $\{\varphi_n\}_n$ as algebraic states on infinite tensor products of matrix algebras. This perspective uncovers fundamental structural properties and classifies MPS into two distinct categories. The first category, known as projective MPS, consists of states that satisfy the consistency condition
\begin{equation}\label{eq_proj}
   \varphi_{n+1} \big|_{\mathcal{B}_{[1,n]}} = \varphi_n
\end{equation}
ensuring that correlations are entirely determined by local finite sequences. Such states retain their correlation structure under thermodynamic limits without requiring additional constraints.

In contrast, non-projective MPS necessitate ergodic quantum channels to define their thermodynamic limit. Utilizing the Markov-Dobrushin inequality for quantum channels \cite{SB24, AcLuSo21}, we rigorously establish the existence of an infinite-volume limiting state and introduce the concepts of ergodic and mixing MPS. This distinction allows us to analyze the interplay between algebraic structure and dynamical properties in infinite tensor network states.
As a case study, we analyze the depolarizing MPS and demonstrate that it lacks a finite correlation structure. Consequently, it cannot be fully described through the correlations derived from $\{\varphi_n\}$, emphasizing the necessity of an alternative, ergodic framework for its characterization.

This work provides a rigorous algebraic foundation for MPS and iMPS, bridging their Heisenberg picture representation with ergodic quantum processes. Our findings contribute to a deeper understanding of phase transitions and long-range entanglement properties in quantum spin systems.

This framework allows for systematic comparisons between MPS and other classes of states, like EMCs \cite{AF05, SBS23, AMO06}, especially in terms of degree of entanglement. One  significant application is embedding the GHZ state within a quasi-local algebra, showcasing how MPS can represent maximal entanglement. This approach strengthens the links between MPS, QMCs, and finitely correlated states, setting the stage for deeper exploration of Markovian and non-Markovian behaviors in quantum systems \cite{Mov21, Mov22, P19}. Future work will focus on examining the ergodic properties and phase transitions of quantum spin chains through this MPS framework.

The paper is structured as follows: Section \ref{Sect_QLA} introduces the necessary preliminaries. In Section \ref{sect_MPSH}, we explore the Heisenberg representation of matrix product states (MPS). Section \ref{Sect_Main} rigorously develops the framework for projective matrix product states, with a particular focus on their application to the GHZ state. Entanglement and mixing properties of MPS are examined in Section \ref{Sect_ergMPS}, along with their thermodynamic limit. Finally, Section \ref{sect_depoMPS} presents illustrative results on ergodic MPS within the context of depolarizing MPS.

\section{Preliminaries }\label{Sect_QLA}

Let \( \mathcal{H} \) denote a finite-dimensional Hilbert space with an orthonormal basis \( \{ |i\rangle \}_{i=1}^{d} \). The algebra of bounded linear operators on \( \mathcal{H} \) is represented by \( \mathcal{B} = \mathcal{B}(\mathcal{H}) \), with the identity operator denoted as \( \id \).  For each $n\in\mathbb{N}$ by $\mathcal{B}_n$ we mean an isomorphic copy of $\mathcal{B}$. We consider the infinite tensor product of these algebras:
\(
\mathcal{B}_{\mathbb{N}} = \bigotimes_{n \in \mathbb{N}} \mathcal{B}_n.
\)
To represent algebras for finite subsystems, we define
\(
\mathcal{B}_{[1,n]} = \mathcal{B}_1 \otimes \mathcal{B}_2 \otimes \cdots \otimes \mathcal{B}_n,
\)
which describes the operator algebra for the first \( n \) subsystems. There exists a natural embedding
\(
\mathcal{B}_{[1,n]} \otimes \id_{n+1} \subset \mathcal{B}_{[1,n+1]},
\)
allowing for the inclusion of the algebra when extending from \( n \) to \( n+1 \) components. The set of local observables, which consists of all algebras acting on finite subsystems, is then defined as
\[
\mathcal{B}_{\mathbb{N}, \text{loc}} = \bigcup_{n} \mathcal{B}_{[1,n]}
\]
By employing the identification \( \mathcal{B}(\mathcal{H})^{\otimes n} \equiv \mathcal{B}(\mathcal{H}^{\otimes n}) \), we can interpret the elements of \( \mathcal{B}_{[1,n]} \) as operators acting on the Hilbert space \( \mathcal{H}^{\otimes n} \).
It is known that the infinite tensor product algebra \( \mathcal{B}_{\mathbb{N}} \) is the closure of the local algebra \( \mathcal{B}_{\mathbb{N}, \text{loc}} \) in the C$^*$-norm. For more details on quasi-local algebras, please refer to \cite{BR}.

The infinite system is described on the tensor product algebra \( \mathcal{B}_{\mathbb{N}} = \bigotimes_{\mathbb{N}} \mathcal{B}(\mathcal{K}) \), with the left shift \( \alpha \) defined as:
\[
\alpha(X_1 \otimes X_2 \otimes \cdots) = \id \otimes X_1 \otimes X_2 \otimes \cdots.
\]

Let $D\in\mathbb{N}$ and $\mathbb{M}_{D}(\mathbb{C})$ be the algebra of $D\times D$ matrices with complex entries.
A \emph{superoperator} \( \Phi: \mathbb{M}_{D}(\mathbb{C}) \to \mathbb{M}_{D}(\mathbb{C}) \) is a linear map acting on the space of \( D \times D \) complex matrices, denoted by \( \mathbb{M}_{D}(\mathbb{C}) \). The \emph{trace} of the superoperator \( \Phi \) is defined as:
\begin{equation}
    \operatorname{Tr}(\Phi) = \sum_{k,l} \big\langle e_{k} \big| \Phi(|e_k\rangle \langle e_l|) \big| e_l \big\rangle,
\end{equation}
where \( \{ e_k \} \) is the standard orthonormal basis of \( \mathbb{C}^{D} \), meaning \( e_k \) is the column vector with 1 in the \( k \)-th position and 0 elsewhere.

The trace of a superoperator quantifies how it acts on the space of bounded operators and is particularly useful in characterizing the structural properties of quantum channels. 

\begin{definition}\label{def-const}
A sequence of states \( (\varphi_n)_n \) is considered projective with respect to the increasing sequence of subalgebras \( (\mathcal{B}_{[1,n]})_n \) if it satisfies the following criteria:
\begin{enumerate}
    \item For each integer \( n \), the map \( \varphi_n \) defines a state on \( \mathcal{B}_{[1,n]} \), meaning that it is a positive linear functional that is normalized as follows:
    \[
    \varphi_n: \mathcal{B}_{[1,n]} \rightarrow \mathbb{C}, \quad \varphi_n(\id) = 1, \quad \varphi_n(b^* b) \geq 0, \quad \forall b \in \mathcal{B}_{[1,n]}.
    \]
    This ensures that \( \varphi_n \) is a well-defined state on the finite region consisting of the first \( n \) subsystems.

    \item The sequence satisfies a consistency condition: When \( \varphi_{n+1} \) is restricted to the algebra \( \mathcal{B}_{[1,n]} \), it coincides with the state \( \varphi_n \). Mathematically, this is expressed as:
    \[
    \varphi_{n+1} \big|_{\mathcal{B}_{[1,n]}} = \varphi_n
    \]
\end{enumerate}

The projectivity condition ensures the compatibility of the states across different system sizes, indicating that the description of any finite subsystem \( [1, n] \) remains consistent when the subsystem is extended to a larger system \( [1, n+k] \).
\end{definition}

The \textit{partial trace} \( \Tr_{n]} \) is defined on the local algebra \( \mathcal{B}_{[1,n+k]} \) as the linear extension of:
\begin{equation}\label{eq_Trn}
\Tr_{n]}(X_{[1,n]} \otimes X_{[n+1, n+k]}) = X_{[1,n]} \Tr(X_{[n+1, n+k]}),
\end{equation}
where \( k \in \mathbb{N} \), \( X_{[1,n]} \in \mathcal{B}_{[1,n]} \), and \( X_{[n+1,n+k]} \in \mathcal{B}_{[n+1,n+k]} \). Here, \( \Tr \) denotes the normalized trace, defined as:
\[
\Tr(X) =   \sum_{i} \langle i | X | i \rangle,
\]
where \( \{|i\rangle\} \) is an orthonormal basis of the Hilbert space.

A linear map \( \Phi: \mathcal{B}(\mathcal{H}) \to \mathcal{B}(\mathcal{K}) \) is said to be \textit{completely positive} if, for any integer \( n \), the extended map \( \Phi \otimes \mathrm{id}_n \) preserves positivity, i.e.,
\[
(\Phi \otimes \mathrm{id}_n)(X) \geq 0, \quad \forall X \in \mathcal{B}(\mathcal{H} \otimes \mathbb{C}^n), \ X \geq 0,
\]
where \( \mathrm{id}_n \) is the identity map on \( \mathbb{C}^n \).

\begin{definition}
A \textit{quantum channel} is a linear map \( \Phi: \mathcal{B}(\mathcal{H}) \to \mathcal{B}(\mathcal{H}') \) that is both completely positive and trace-preserving (CPTP). For \( \mathcal{H} = \mathcal{H}' \), the channel acts within a single Hilbert space.

A \textit{Markov operator} is the dual of a quantum channel with respect to the Hilbert-Schmidt inner product. If the quantum channel \( \Phi \) is represented in terms of Kraus operators \( \{K_i\}_{i \in I} \) as:
\[
\Phi(\rho) = \sum_{i \in I} K_i \rho K_i^\dagger, \quad \rho \in \mathcal{B}(\mathcal{H}),
\]
the corresponding Markov operator \( \mathcal{T} \) acts on observables \( A \in \mathcal{B}(\mathcal{H}) \) as:
\[
\mathcal{T}(A) = \sum_{i \in I} K_i^\dagger A K_i.
\]
Markov operator
\end{definition}
Quantum Markov Chains (QMCs) and Finitely Correlated States (FCS) both describe quantum states with finite-range correlations.

A QMC is defined on the algebra $\mathcal{B}_{\mathbb{N}}$ using transition expectations $\mathcal{E}_n: \mathcal{B}_n \otimes \mathcal{B}_{n+1} \to \mathcal{B}_n$:
\[
\varphi(X_1 \otimes \cdots \otimes X_n) = \phi_0\big(\mathcal{E}_0\big(X_0 \otimes \mathcal{E}_1\big(\cdots \mathcal{E}_n\big(X_n \otimes \id_{n+1}\big)\big)\big)\big)\big)
\]
This ensures finite correlation length and Markovian dynamics.

FCS are translation-invariant states on $\mathcal{B}(\mathcal{H})^{\otimes \mathbb{N}}$, characterized by an auxiliary space $\mathcal{K}$, a CPIP map $\mathbb{E}: \mathcal{B}(\mathcal{H}) \otimes \mathcal{B}(\mathcal{K}) \to \mathcal{B}(\mathcal{K})$, and a state $\rho$ on $\mathcal{B}(\mathcal{K})$. Local expectations are given by:
\[
\varphi(X_1 \otimes \cdots \otimes X_n) = \rho\big(\mathbb{E}_{X_1} \circ \cdots \circ \mathbb{E}_{X_n}(\id_{\mathcal{K}})\big).
\]

Both structures use identity-preserving transition expectations to enforce finite correlation properties. By defining $\mathcal{E}(X_n\otimes X_{n+1}) = \mathbb{E}_{X_n}(\mathbb{E}_{X_{n+1}}(\id_{\mathcal{K}}))$, QMCs can be seen as a special case of FCS. Moreover, both satisfy the projectivity condition \eqref{eq_proj}, ensuring consistency in  the sense of Definition \ref{def-const}.

\section{ Matrix Product States in  Quantum Spin Chains}\label{sect_MPSH}
 Consider a sequence of  $D\times D$ matrices \( \{A_{i_k}^{[k]}\}_{0\le i_k \le d} \),  associated with the \( k \)-th site, and \( |i_k\rangle \) denotes the local basis states. corresponding to the local basis of the Hilbert space \( \mathcal{H} \). Matrix Product States (MPS) are an efficient way to represent quantum many-body states, particularly in one-dimensional lattice systems. An MPS is expressed as:
 \begin{equation}\label{psn_n}
|\psi_n\rangle = \sum_{i_1, i_2, \ldots, i_n} \operatorname{Tr}(A_{i_1}^{[1]} A_{i_2}^{[2]} \cdots A_{i_n}^{[n]}) |i_1, i_2, \ldots, i_n\rangle,
\end{equation}
where \( A_{i_k}^{[k]} \)  The bond dimension $D$ of the matrices \( A_{i_k}^{[k]} \) determines the degree of entanglement captured by the state, with larger bond dimensions allowing for greater entanglement.

 These matrices satisfy a   gauge condition:
\begin{equation}\label{eq_sumAkAkd1}
\sum_{i \in I} A_i^{[k]} A_i^{[k]\, \dagger} = \id_{D}
\end{equation}

 The normalization constant is computed as:
\[
\mathcal{N}(n) = \sum_{i_1, i_2, \ldots, i_n} \big|\operatorname{Tr}\big(A_{i_1}^{[1]} A_{i_2}^{[2]} \cdots A_{i_n}^{[n]} \big)\big|^2
\]

Matrix Product States (MPS) can be naturally formulated in the Heisenberg picture, where the focus shifts from the evolution of states to the evolution of observables. In this framework, for an observable \( X \) acting on the first \( n \) sites, its expectation value in the Heisenberg picture is given by:

\begin{equation}\label{eq_varphi}
\varphi_n(X) = \frac{\big\langle \Psi_{n+1} \big| X \otimes \id_{n+1} \big| \Psi_{n+1} \big\rangle}{\mathcal{N}(n+1)},
\end{equation}

where \( X \in \mathcal{B}_{[1,n]} \) is an operator acting on the first \( n \) sites, and the identity operator \(\id_{n+1}\) acts on the \((n+1)^{th}\) site. This formulation allows us to define a sequence of states \(\{\varphi_n\}_n\) that captures the statistical and correlation properties of the MPS as the system size grows.

A fundamental question in the study of MPS is whether this sequence has a well-defined limit as \( n \to \infty \), leading to an infinite-volume MPS. The existence of such a limiting state characterizes the asymptotic behavior of the system, revealing long-range correlations and ergodic properties.

\begin{definition}
The sequence of states \(\{\varphi_n\}_n\) given by (\ref{eq_varphi}) defines an infinite-volume matrix product state on the spin chain if the limit

\begin{equation}\label{eq-philim}
\varphi = \lim_{n \to \infty} \varphi_n
\end{equation}

exists and defines a state on the algebra \( \mathcal{B}_{\mathbb{N}} \). This limiting state \(\varphi\) describes the thermodynamic behavior of the system, providing a rigorous framework for studying its large-scale structure.
\end{definition}

For each $n$   the map
\begin{equation}\label{eq-Phin}
  \Phi_n(M) = \sum_{i} A_{i}^{[n]\, \dagger} M A_{i}^{[n]}; \qquad \forall M\in\mathbb{M}_D(\mathbb{C})
\end{equation}
is a CP map from $\mathbb{M}_D(\mathbb{C})$ into itself. If the gauge condition  A recursive composition of  $\{\Phi_k\}_k$ leads to
$$
\Phi_n\circ\cdots\circ  \Phi_1(M)= \sum_{i_1, i_2,\ldots, i_n} A_{i_n}^{[k]\, \dagger}\cdots A_{i_1}^{[1]\, \dagger} M  A_{i_1}^{[1]}\cdots A_{i_n}^{[k]}; \qquad \forall M\in\mathbb{M}_D(\mathbb{C})
$$
One can check that
$$
 \mathcal{N}(n) = \Tr(\Phi_n\circ\cdots\circ  \Phi_1)
$$
 For every observable $X\in\mathcal{B}_{[m,n]}$, we define the  associated linear map $\hat{X}$ from $\mathbb{M}_D(\mathbb{C})$ into itself by:
 \begin{equation}\label{eq-hatX}
   \hat{X}(M) = \sum_{\substack{i_m,\dots,i_{n} \\ j_m,\dots,j_{n}}}\big\langle i_m,\ldots, i_n \big| X \big| j_m,\ldots, j_n\big\rangle A_{i_n}^{[n]\, \dagger}\cdots A_{i_m}^{[m]\, \dagger} M  A_{j_m}^{[m]}\cdots A_{j_n}^{[n]}
 \end{equation}
 For $X = \id_{[m,n]}$ one has $\hat{\id}_{[m,n]}(M) = \Phi_{n}\circ\cdots\circ \Phi_m(M)$.  For two observables $X$ and $Y$ localized in separated region one can check that
 $$
 \hat{X}\circ \hat{Y} = \widehat{X\otimes Y}
 $$

One on the maim purposes is to find the local correlations of the limiting state of $(\varphi_n)_n$. Such as the case of quantum Markov chains and fimitely correlated states.

The Heisenberg picture is particularly useful for studying the thermodynamic limit of MPS, as it provides insight into long-range correlations, decay of off-diagonal terms, and ergodic properties. Moreover, this framework enables a natural connection between MPS and other representations of quantum many-body systems, such as quantum Markov chains and finitely correlated states.

\section{Projective Matrix product states}\label{Sect_Main}

In this section, we establish a rigorous construction for a state on an infinite tensor product of matrix algebras, formulated using a sequence of matrix product operators. This construction allows us to define a state with preserved correlation structure across an infinite lattice, extending naturally to applications in quantum information and many-body physics.
  We aim to find a suitable condition that grantees the infinite limit of the sequence of states $\{\varphi_n\}_n$  given by (\ref{eq_varphi}) and study its  correlation functions and later its  ergodic properties.\\
To proceed, we introduce the following technical lemma, which will be useful in subsequent calculations.

\begin{lemma}
   For any \( n, k \in \mathbb{N} \) and for any index sequences \( i_1, i_2, \dots, i_{n+k} \) and \( j_1, j_2, \dots, j_{n+k} \), the following identity holds:
  \begin{equation}\label{eq_trkl}
  \overline{\operatorname{Tr}\left( A_{i_1}^{[1]} A_{i_2}^{[2]} \cdots A_{i_{n+k}}^{[n+k]} \right)} \operatorname{Tr}\left( A_{j_1}^{[1]} A_{j_2}^{[2]} \cdots A_{j_{n+k}}^{[n+k]} \right)
  \end{equation}
  $$
   = \operatorname{Tr}\left[ \left( {A_{i_{n}}^{[n]}}^{\dagger} \cdots {A_{i_{1}}^{[1]}}^{\dagger} \otimes A_{j_1}^{[1]} \cdots A_{j_{n}}^{[n]} \right) \left(
   {A_{i_{n+k}}^{[n+k]}}^{\dagger} \cdots {A_{i_{n+1}}^{[n+1]}}^{\dagger} \otimes A_{i_{n+1}}^{[n+1]} \cdots A_{i_{n+k}}^{[n+k]} \right) \right]
  $$
\end{lemma}
\begin{proof}
  To prove this identity, we will apply the conjugate, cyclic, and tensor product properties of the trace

  Consider the left-hand side of Equation \eqref{eq_trkl}:
  \[
  \overline{\operatorname{Tr}\left( A_{i_1}^{[1]} A_{i_2}^{[2]} \cdots A_{i_{n+k}}^{[n+k]} \right)} \operatorname{Tr}\left( A_{j_1}^{[1]} A_{j_2}^{[2]} \cdots A_{j_{n+k}}^{[n+k]} \right)
  \]
  Taking the conjugate of the trace term, we have
  \[
  \overline{\operatorname{Tr}\left( A_{i_1}^{[1]} A_{i_2}^{[2]} \cdots A_{i_{n+k}}^{[n+k]} \right)} = \operatorname{Tr}\left( {A_{i_{n+k}}^{[n+k]}}^{\dagger} \cdots {A_{i_1}^{[1]}}^{\dagger} \right)
  \]
  Using this, the expression becomes
  \[
  \operatorname{Tr}\left( {A_{i_{n+k}}^{[n+k]}}^{\dagger} \cdots {A_{i_1}^{[1]}}^{\dagger} \right) \operatorname{Tr}\left( A_{j_1}^{[1]} A_{j_2}^{[2]} \cdots A_{j_{n+k}}^{[n+k]} \right)
  \]
  By splitting the sequences at index \( n \) in the conjugated  terms, we get
  \[
  \operatorname{Tr}\left( {A_{i_{n+k}}^{[n+k]}}^{\dagger} \cdots {A_{i_1}^{[1]}}^{\dagger} \right) = \operatorname{Tr}\left( {A_{i_{n}}^{[n]}}^{\dagger} \cdots {A_{i_1}^{[1]}}^{\dagger} {A_{i_{n+k}}^{[n+k]}}^{\dagger} \cdots {A_{i_{n+1}}^{[n+1]}}^{\dagger} \right)
  \]
  Finally, by using the  tensor property of the trace, we can combine the two traces on the left-hand side of Equation \eqref{eq_trkl} under a single trace as
  \[
  \operatorname{Tr}\left[ \left( {A_{i_{n}}^{[n]}}^{\dagger} \cdots {A_{i_1}^{[1]}}^{\dagger} \otimes A_{j_1}^{[1]} \cdots A_{j_{n}}^{[n]} \right) \left(
   {A_{i_{n+k}}^{[n+k]}}^{\dagger} \cdots {A_{i_{n+1}}^{[n+1]}}^{\dagger} \otimes A_{i_{n+1}}^{[n+1]} \cdots A_{i_{n+k}}^{[n+k]} \right) \right]
  \]
  This completes the proof.
\end{proof}

\begin{theorem}\label{thm_main}
Let \( \{A_{i_n}^{[n]}\}_{n} \) be a family of \( D \times D \) matrices, and let \( \Phi_n \) be defined as in Equation~(\ref{eq-Phin}). Suppose that for all \( n \in \mathbb{N} \), the operator identity holds:
\begin{equation}\label{eq_operator_identity}
    \sum_{j} A_{j}^{[n+1] \dagger} A_{i}^{[n] \dagger} \otimes A_{i}^{[n]} A_{j}^{[n+1]} = A_{i}^{[n] \dagger} \otimes A_{i}^{[n]}
\end{equation}
Then, the sequence \( \{\Phi_n\}_n \) satisfies:
\begin{equation}\label{eq-Phnn1}
    \operatorname{Tr}(\Phi_{n+1} \circ \Phi_{n} \circ \Phi) = \operatorname{Tr}(\Phi_n \circ \Phi), \quad \forall n \in \mathbb{N}
\end{equation}
for every CP map \( \Phi: \mathbb{M}_D(\mathbb{C}) \to \mathbb{M}_D(\mathbb{C}) \).
Moreover the sequence of states \(\{\varphi_n\}_n\) given by (\ref{eq_varphi}) is projective, and  for any \( N\in\mathbb{N} \) and local observable \( X \in \mathcal{B}_{[1,N]} \), the thermodynamic limit:
\begin{equation}\label{eq_state_limit}
    \varphi(X) := \lim_{n\to\infty} \frac{\langle \Psi_n |X\otimes \id| \Psi_n \rangle}{\langle \Psi_n |  \Psi_n \rangle} = \frac{\operatorname{Tr}(\Phi_{N+1} \circ \hat{X})}{\operatorname{Tr}(\Phi_1)}
\end{equation}
defines a quantum state \( \varphi \) on \( \mathcal{B}_{\mathbb{N}} \).
\end{theorem}
\begin{proof}
Let \( \Phi(M) = \sum_{\ell \in J} K_{\ell}^{\dagger} M K_{\ell} \) be a completely positive (CP) map on \( \mathbb{M}_D(\mathbb{C}) \), where \( J \) is a finite index set.

Utilizing the definition of \( \Phi_n \) from (\ref{eq-Phin}), we obtain:
\begin{equation*}
    \Phi_{n+1} \circ \Phi_n \circ \Phi(M) = \sum_{i,j,\ell} A_{j}^{[n+1] \dagger} A_{i}^{[n] \dagger} K_{\ell}^{\dagger} M K_{\ell} A_{i}^{[n]} A_{j}^{[n+1]}
\end{equation*}
This recursive formulation highlights the structured evolution of the CP map under sequential compositions, preserving its complete positivity and operational consistency.

Taking the trace and using linearity, we get:
\begin{equation*}
    \Tr(\Phi_{n+1} \circ \Phi_n \circ \Phi) = \sum_{i,j,\ell} \Tr\big(A_{j}^{[n+1] \dagger} A_{i}^{[n] \dagger} K_{\ell}^{\dagger}\big) \Tr\big(K_{\ell} A_{i}^{[n]} A_{j}^{[n+1]} \big).
\end{equation*}

Using the cyclic property of the trace and the identity \eqref{eq_trkl}, we simplify:
\begin{equation*}
    \Tr(\Phi_{n+1} \circ \Phi_n \circ \Phi) = \sum_{i,\ell} \Tr\big( (K_{\ell}^{\dagger} \otimes K_{\ell}) \sum_{j} A_{j}^{[n+1] \dagger} A_{i}^{[n] \dagger} \otimes A_{i}^{[n]} A_{j}^{[n+1]} \big).
\end{equation*}

Applying the consistency condition \eqref{eq_operator_identity}, we further reduce:
\begin{equation*}
    \Tr(\Phi_{n+1} \circ \Phi_n \circ \Phi) = \sum_{i,\ell} \Tr\big( (K_{\ell}^{\dagger} \otimes K_{\ell}) (A_{i}^{[n] \dagger} \otimes A_{i}^{[n]}) \big).
\end{equation*}

Using the decomposition property of the trace, we finally obtain:
\begin{equation*}
    \Tr(\Phi_{n+1} \circ \Phi_n \circ \Phi) = \Tr(\Phi_n \circ \Phi),
\end{equation*}

 Next, we consider a local observable \( X \in \mathcal{B}_{[1,N]} \). For any \( k \), the expected value \( \big\langle \Psi_{N+k+1} \big|X\otimes \id\big| \Psi_{N+k+1}\big\rangle \) is given by
\begin{align*}
\big\langle \Psi_{N+k+1} \big|X\otimes \id\big| \Psi_{N+k+1}\big\rangle &= \sum_{\substack{i_1, i_2, \dots, i_{N+k} \\ j_1, j_2, \dots, j_{N+k+1}}}
\quad   \langle i_1, i_2, \dots, i_{N+k+1} | X \otimes \id | j_1, j_2, \dots, j_{N+k+1} \rangle\\
&\qquad \times \overline{\Tr\left( A_{i_1}^{[1]}   \cdots A_{i_{N+k+1}}^{[N+k+1]} \right)} \Tr\left( A_{j_1}^{[1]}   \cdots A_{j_{N+k+1}}^{[N+k+1]} \right) \\
&= \sum_{\substack{i_1, i_2, \dots, i_{N} \\ j_1, j_2, \dots, j_{N}}}\sum_{i_{N+1},\dots, i_{N+k}}\qquad \langle i_1, i_2, \dots, i_{N} | X   | j_1, j_2, \dots, j_{N} \rangle\\&\qquad\times\Tr\left(A_{i_{N+k}}^{[N+k]\, \dagger}\cdots  A_{i_{N+1}}^{[N+1]\, \dagger}A_{i_N}^{[N]\, \dagger}\cdots  A_{i_1}^{[1]\, \dagger}\right) \\ &\qquad\times\Tr\left( A_{j_1}^{[1]}\cdots A_{j_N}^{[N]} A_{i_{N+1}}^{[N+1]}\cdots A_{i_{N+k}}^{[N+k]} \right) \\
\end{align*}
Considering that for every integer $n$ one has
\begin{align*}
 \overline{\Tr\left( A_{i_1}^{[1]}   \cdots A_{i_{n}}^{[n]} \right)} \Tr\left( A_{j_1}^{[1]}   \cdots A_{j_{n}}^{[n]} \right) &= \sum_{\alpha,\beta=1}^{D}\big\langle e_{\alpha} \big|A_{i_{n}}^{[n]\,\dagger}\cdots A_{i_1}^{[1]\,\dagger}\big|e_{\alpha}\big\rangle\big\langle e_{\beta}| A_{j_1}^{[1]}   \cdots A_{j_{n}}^{n}| e_{\beta}\big\rangle\\
 & = \sum_{\alpha,\beta=1}^{D}\Tr\big(\Phi_{i_n,j_n}\circ\cdots\circ \Phi_{i_1,j_1}\big)
\end{align*}
where $\Phi_{i_m,j_m}: M\in\mathbb{M}_{D}(\mathbb{C}) \mapsto A_{i_{m}}^{\dagger} M A_{j_{m}}$ and the expression of $\hat{X}$ given in (\ref{eq-hatX}) one gets
\[
\big\langle \Psi_{N+k+1} \big|X\otimes \id\big| \Psi_{N+k+1}\big\rangle = \Tr\big(\Phi_{N+k}\circ\cdots\circ\Phi_{N+1}\circ\hat{X}\big)
\]
For a local observable \( X \in \mathcal{B}_{[1,N]} \), we express the expectation value as:
\begin{equation*}
    \varphi_{N+k}(X) = \frac{\Tr\big(\Phi_{N+k} \circ \cdots \circ \Phi_{N+1} \circ \hat{X}\big)}{\mathcal{N}_{N+k}}.
\end{equation*}

Using the iterative property of \eqref{eq-Phnn1}, we obtain:
\begin{equation*}
    \Tr\big(\Phi_{N+k} \circ \cdots \circ \Phi_{N+1} \circ \hat{X} \big) = \Tr\big(\Phi_{N+1} \circ \hat{X}\big).
\end{equation*}

In particular, the normalization factor $\mathcal{N}_{N+k}= \Tr\big(\Phi_{N+k} \circ \cdots \circ \Phi_{1}\big)$ simplifies as:
\begin{equation*}
    \mathcal{N}_{N+k} = \mathcal{N}_{N+1} = \mathcal{N}_1 =  \Tr(\Phi_1),
\end{equation*}
we conclude:
\begin{equation*}
    \varphi_{N+k}(X) = \varphi_{N} = \frac{\Tr(\Phi_{N+1} \circ \hat{X})}{\Tr(\Phi_1)}.
\end{equation*}

This confirms that the sequence \( \{\varphi_n\}_n \) remains consistent, guaranteeing that the limit \eqref{eq_state_limit} exists for all local observables \( X \). Consequently, \( \varphi \) defines a well-posed state on \( \mathcal{B}_{\mathbb{N}, \text{loc}} \), naturally extending to a state on the full algebra \( \mathcal{B}_{\mathbb{N}} \).
\end{proof}
\begin{remark}
This result offers a structured approach to defining states on an infinite tensor product space, where the matrix product structure satisfies clear consistency conditions. In particular, Equation  \eqref{eq_operator_identity} aligns with consistency conditions in  the sense of Definition \ref{def-const} studied in the context of QMCs  on trees, especially for Pauli models associated with phase transitions \cite{MBS161, MS19, MSHA22}. This parallel opens a potential pathway for extending Theorem \ref{thm_main} to  MPS on tree structures, with possible applications in quantum information theory and statistical mechanics.
\end{remark}

\subsection{Application to the GHZ State}\label{Sect_GHZ}

In this section, we apply the main theorem to the GHZ state, a prototypical example of a maximally entangled state, and construct the necessary Kraus operators \( K_i \) to satisfy the theorem’s conditions. We aim to capture the entanglement structure and the non-Markovian behavior of the GHZ state using the formalism of quantum Markov chains.

Let \( \mathcal{H} = \mathbb{C}^2 \) be a two-dimensional Hilbert space with the standard orthonormal basis given by
\[
|0\rangle = \begin{pmatrix} 1 \\ 0 \end{pmatrix}, \quad |1\rangle = \begin{pmatrix} 0 \\ 1 \end{pmatrix}
\]
The algebra of bounded operators on \( \mathcal{H} \) is denoted by \( \mathcal{B} = \mathcal{B}(\mathcal{H}) \), and the quasi-local algebra for an infinite sequence of qubits is defined by the infinite tensor product
\[
\mathcal{B}_{\mathbb{N}} = \bigotimes_{n \in \mathbb{N}} \mathcal{B}
\]
The \( n \)-qubit GHZ state is defined as
\[
|\text{GHZ}\rangle_n = \frac{1}{\sqrt{2}} \left( |0\rangle^{\otimes n} + |1\rangle^{\otimes n} \right)
\]
where \( |0\rangle^{\otimes n} \) and \( |1\rangle^{\otimes n} \) are the states with all qubits in \( |0\rangle \) and \( |1\rangle \), respectively.

We define a set of matrices \( \{A_i^{[n]}\}_{i=1,2} \) for each site \( n \) to represent the local operations associated with each qubit. These matrices will satisfy the conditions required by the main theorem. Define
\begin{equation} \label{eq_A1}
    A_1^{[n]} =  \begin{pmatrix} 1 & 0 \\ 0 & 0 \end{pmatrix}, \quad A_2^{[n]} =  \begin{pmatrix} 0 & 0 \\ 0 & 1 \end{pmatrix}
\end{equation}

These matrices act as diagonal projections and encode the information necessary to distinguish between the two components of the GHZ state.
The matrices $\{ A_{i_n}^{[n]}; n\ge 1, i_n =1,2\}$ satisfy  the consistency condition (\ref{eq_operator_identity}):
\[
\sum_{j=1}^{2} \left( A_{j}^{[n+2]} \right)^{\dagger} \left( A_{i}^{[n+1]} \right)^{\dagger} \otimes A_{i}^{[n+1]} A_{j}^{[n+2]} = \left( A_{i}^{[n+1]} \right)^{\dagger} \otimes A_{i}^{[n+1]}
\]

The quantum channel associated with translation invariant MPS $\{A^{[n]}_{i}\}_i$ is defined on $M = \begin{pmatrix} a & b \\ c & d \end{pmatrix}\in \mathbb{M}_2(\mathbb{C})$ by
$$
\Phi(M) = \begin{pmatrix} a & 0 \\ 0 & d \end{pmatrix}
$$
From Theorem \ref{thm_main}, the limiting state  \( \varphi_{GHZ}  \) according (\ref{eq_state_limit})   exists and its local correlations are given by
\begin{eqnarray*}
    \varphi_{GHZ}(X) &=\sum_{\ell=1}^{2} \sum_{\substack{i_1, i_2, \dots, i_N \\ j_1, j_2, \dots, j_N}} \overline{\operatorname{Tr}\left( A_{i_1}^{[1]} \cdots  A_{i_N}^{[N]}A_{\ell}^{[N+1]} \right)} \operatorname{Tr}\left( A_{j_1}^{[1]} \cdots  A_{j_N}^{[N]} A_{\ell}^{[N+1]} \right)\\
&\times \langle i_1, i_2, \dots, i_n | X | j_1, j_2, \dots, j_n \rangle,
\end{eqnarray*}
The  matrices defined in \eqref{eq_A1}   satisfy:
\[
\operatorname{Tr}\left( A_{i_1}^{[1]} A_{i_2}^{[2]} \cdots A_{i_N}^{[N]} A_{\ell}^{[N+1]}\right) = \frac{1}{\sqrt{2}} \delta_{i_1, i_2, \dots, i_N, \ell}
\]
where
\[
\delta_{i_1, i_2, \dots, i_N} = \begin{cases}
1 & \text{if } i_1 = i_2 = \cdots = i_N =\ell \\
0 & \text{otherwise,}
\end{cases}
\]
It follows that
\begin{equation} \label{eq_phi_GHZ}
\varphi_{GHZ}(X_1 \otimes X_2 \otimes \cdots \otimes X_N) = \frac{1}{2} \sum_{\ell=1}^{2} x_{1;\ell\ell} x_{2;\ell\ell} \cdots x_{N;\ell\ell}
\end{equation}
where \( x_{m;ij} = \langle i | X_m | j \rangle \).

\section{Ergodic and Mixing translation invariant  MPS }\label{Sect_ergMPS}
In this section we introduce the notions of ergodicity and mixing of translation invariant matrix product states based on the formalism of quantum channels. Denote by \(\mathbb{M}_D(\mathbb{C})_+\) the set of all positive semidefinite matrices in \(\mathbb{M}_D(\mathbb{C})\), given by
\[\mathbb{M}_D(\mathbb{C})_+ = \left\{ A \in \mathbb{M}_D(\mathbb{C}) \mid A^\dagger = A, \, A \geq 0 \right\}\]
This set consists of all \(D \times D\) complex Hermitian matrices that are positive semidefinite. Denote by \(\mathfrak{S}_D\) the set of all density matrices, defined as:
\[
\mathfrak{S}_D = \left\{ \rho \in \mathbb{M}_D(\mathbb{C}) \mid \rho^\dagger = \rho, \, \rho \geq 0, \, \operatorname{Tr}(\rho) = 1 \right\}\subset \mathbb{M}_D(\mathbb{C})_+
\]
This set comprises all \( D \times D \) complex matrices that are self-adjoint, positive semi-definite, and possess unit trace, characterizing valid quantum states.

\begin{definition}
A quantum channel $\Phi: \mathbb{M}_D(\mathbb{C}) \to \mathbb{M}_D(\mathbb{C})$ is \emph{ergodic} if it has a unique invariant state $\rho_*$ within the quantum state space $\mathfrak{S}_D$.

Additionally, $\Phi$ is \emph{mixing} if, for any initial state $\rho \in \mathfrak{S}_D$, repeated application of $\Phi$ drives $\rho$ toward $\rho_*$:
\[
\lim_{n \to \infty} \Phi^n(\rho) = \rho_*.
\]
This property ensures that regardless of the starting state, the system eventually stabilizes at $\rho_*$.
\end{definition}

\begin{definition}
A translation-invariant Matrix Product State (MPS), defined by the tensor set $\{A_i\}_i$, is \emph{ergodic} if the associated quantum channel
\[
\Phi: M \mapsto \sum_{i} A_{i}^{ \dagger} M A_{i}
\]
is ergodic. Similarly, the MPS is \emph{mixing} if $\Phi$ is mixing, ensuring that iterates of the channel converge to a unique stationary state.
\end{definition}

These definitions highlight the fundamental role of ergodicity and mixing in quantum dynamics. Ergodicity guarantees that the system possesses a unique long-term behavior, while mixing strengthens this by ensuring convergence from any initial condition. In the context of MPS, these properties determine the stability and uniformity of correlations, making them crucial for understanding entanglement structure and thermodynamic behavior. Ergodic and mixing MPS are particularly relevant in quantum many-body physics, where they describe systems reaching equilibrium under repeated interactions. The connection between these properties and the underlying tensor structure provides a powerful framework for analyzing quantum information processes and statistical mechanics.

 We use the same notations as the previous sections.
Any matrix $M \in \mathbb{M}_D(\mathbb{C})$ can be uniquely decomposed into its real and imaginary components
\( M = \Re(M) + i \Im(M)\), where \( \Re(M) = \frac{1}{2} (M + M^*), \quad \Im(M) = \frac{1}{2i} (M - M^*).
 \) Both $\Re(M)$ and $\Im(M)$ are self-adjoint. Furthermore, any real self-adjoint matrix $M$ can be expressed as the difference of two positive matrices
\(    M = M_+ - M_-, \quad M_+ M_- = 0\), where $M_+$ and $M_-$ have disjoint support projections. We define the trace norm of $M$ as
\(\|M\|_1 = \operatorname{Tr}(M_+) + \operatorname{Tr}(M_-).\)

In the next lemma, we extend this norm to the total variation norm $\|\cdot\|_{TV}$ on $\mathbb{M}_D(\mathbb{C})$, generalizing $\|\cdot\|_1$ to the full algebra. This norm plays a crucial role in our subsequent analysis. The result stated below was established in \cite{AcLuSo21}.

\begin{lemma}
Define the total variation norm for $M \in \mathbb{M}_D(\mathbb{C})$ as
\begin{equation}\label{df-q-tot-var-nrm}
\|M\|_{TV} := \operatorname{Tr} \big(\Re(M)_+ + \Re(M)_-\big) + \operatorname{Tr} \big(\Im(M)_+ + \Im(M)_-\big)
\end{equation}
This defines a norm on $\mathbb{M}_D(\mathbb{C})$ when viewed as a real vector space.
\end{lemma}

 The \emph{infimum} (or greatest lower bound) of a family of operators $\{M_i\}_{i \in I}$ is defined as the largest operator $B$ such that $B \leq M_i$ for all $i \in I$. Formally, this is expressed as:
\begin{equation*}\label{definition-inf-pos-ops}
\inf_{i \in I} M_i = \sup \{ B \in \mathbb{M}_D(\mathbb{C}) \mid B \leq M_i \text{ for all } i \in I \},
\end{equation*}
where the inequality $B \leq M_i$ means that $M_i - B$ is a positive operator.

For a quantum channel $\Phi:  \mathbb{M}_D(\mathbb{C}) \to  \mathbb{M}_D(\mathbb{C})$, we define its \textit{Markov-Dobrushin constant} as:
\begin{equation}\label{definition-kappa-q}
\kappa_{\Phi} := \inf \left\{ \Phi(\mathbf{\xi} \mathbf{\xi}^\dagger) \mid \|\mathbf{\xi}\| =1 \right\} \in \mathbb{M}_D(\mathbb{C})_+
\end{equation}
This quantity plays a fundamental role in quantifying the ergodic behavior of quantum channels and is a direct extension of the classical Markov-Dobrushin inequality, which is central in analyzing mixing properties of Markov processes \cite{AcLuSo21}. The following result establishes a quantum Markov-Dobrushin inequality, highlighting the stability and convergence properties of $\Phi$.

\begin{theorem}\label{thm_MD}
Let $\Phi:  \mathbb{M}_D(\mathbb{C}) \to  \mathbb{M}_D(\mathbb{C})$ be a quantum channel. Then, for any $\rho, \sigma \in \mathfrak{S}_{D}$,
\begin{equation}\label{q-MD-ineq}
    \|\Phi(\rho) - \Phi(\sigma)\|_{\text{TV}} \leq (1 - \operatorname{Tr} \kappa_{\Phi}) \|\rho - \sigma\|_{\text{TV}}
\end{equation}
If $\kappa_{\Phi} > 0$, then $\Phi$ has a unique fixed point $\rho_*$ and satisfies
\begin{equation}\label{eq-cv-rate}
    \|\Phi^n(\rho) - \rho_{*}\| \leq 2e^{-n \theta_{\Phi}}, \quad \forall \rho \in \mathfrak{S}(\mathcal{H}),
\end{equation}
where
\begin{equation}
    \theta_{\Phi} = -\ln(1 - \operatorname{Tr} \kappa_{\Phi}) >0.
\end{equation}
\end{theorem}
\begin{proof}
  See \cite{AcLuSo21, SB24}.
\end{proof}

The parameter $\theta_{\Phi}$ governs the speed of convergence, with larger values indicating  faster mixing. In particular, when $\kappa_{\Phi}$ is strictly positive, the channel rapidly forgets initial conditions, making it crucial for quantum information processing, thermalization, and quantum memory effects. This result provides a rigorous foundation for analyzing the stability and efficiency of quantum channels in applications such as quantum error correction, quantum Markov chains, and quantum communication protocols.

In  the translation-invariant systems, the matrices \( A_{i_k}^{[k]} \) does not depend on the site \( k \).  In this case the quantum channels $\Phi_n$ given by (\ref{eq-Phin}) coincide with the quantum channel
\begin{equation}\label{eq_M_nonTI}
\Phi(M) = \sum_{i \in I} A_i  M A_i^{\dagger},
\end{equation}
where \( K_i^{[k]} = A_i^{[k]*} \), defines a completely positive trace-preserving (CPTP) quantum channel acting on \( \mathbb{M}_d(\mathbb{C}) = \mathcal{B}(\mathcal{H}) \), where \( \mathcal{H} \) is a \( d \)-dimensional Hilbert space.

Consider the map
\begin{equation}\label{eq-Phi}
  \Phi(M) = \sum_{i}A_{i}^{\dagger}MA_{i},\qquad M\in \mathbb{M}_D(\mathbb{C})
\end{equation}

The dynamics of the non-translation-invariant quantum system associated with the MPS is governed by the sequence of states \( \{\varphi_n\}_{n} \).

For any local observable \( X \in \mathcal{B}_{[1,N]} \), the associated map is:
\begin{equation}\label{eq_Xhat_nonTI}
\hat{X}(M) = \sum_{\substack{i_1, \ldots, i_N \\ j_1, \ldots, j_N}} \big\langle i_1, \ldots, i_N \big| X \big| j_1, \ldots, j_N \big\rangle A_{i_N}^{[N]\dagger} \cdots A_{i_1}^{[1]\dagger} M A_{j_1}^{[1]} \cdots A_{j_N}^{[N]}.
\end{equation}

Given the gauge condition in \eqref{eq_sumAkAkd1}, if \( X = \id_{[1,N]} \), the map \( \hat{X} \) coincides with the quantum channel \( \Phi^N \).

\begin{lemma}\label{lem_limPhin}
Consider a quantum channel \( \Phi : \mathbb{M}_D(\mathbb{C}) \to \mathbb{M}_D(\mathbb{C}) \). The There exists a unique density matrix \( \rho_* \in \mathfrak{S}_D \) that satisfies the following asymptotic property:
for any \( A \in \mathbb{M}_D(\mathbb{C}) \), the iterated application of \( \Phi \) converges to
\begin{equation}\label{eq_limMn}
    \lim_{n \to \infty} \Phi^{n} = \Phi_{*}: M\in \mathbb{M}_D(\mathbb{C}) \mapsto \Phi_{*}(M) =  \mathrm{Tr}(M)\rho_*
\end{equation}
Here, \( \rho_* \) represents the unique stationary state of \( \Phi \), serving as its fixed point.
\end{lemma}
\begin{proof}
See \cite{SN24, W12}.
\end{proof}
\begin{theorem}\label{thm_ergMPS}
Under the given notations, if the Markov-Dobrushin constant \(\kappa_{\Phi}\) of the quantum channel \(\Phi\), defined by (\ref{eq-Phi}), satisfies \(\kappa_{\Phi} > 0\), then for any local observable \( X \in \mathcal{B}_{[1,N]} \), the sequence \(\{\varphi_n\}\) converges to a unique limiting state \(\varphi\), given by:
\begin{equation}\label{lim-phi-hilbert}
    \varphi(X) = \sum_{\alpha,\beta =1}^{D}e_{\alpha}^{\dagger} \rho_{*}e_{\beta} \Tr\big(\hat{X}( e_{\alpha}e_{\beta}^{\dagger})\big),
\end{equation}

where \( \rho_* \) is the unique invariant state of \( \Phi \) within the set of density operators \( \mathfrak{S}_{D} \).
\end{theorem}
\begin{proof}
For any observable $X \in \mathcal{B}_{[1,N]}$, we express its expectation under iterated channel dynamics as:
\[
\Tr\big(\Phi^n\circ \hat{X} \big) = \sum_{\alpha, \beta =1}^{D} e_{\alpha}^{\dagger} \Phi^n\big(\hat{X}(e_{\alpha}e_{\beta}^{\dagger})\big)e_{\beta}.
\]
Since $\kappa_{\Phi} > 0$ by assumption, Theorem \ref{thm_MD} ensures that $\Phi$ is mixing. Applying Lemma \ref{lem_limPhin}, the limiting channel $\Phi_*$ satisfies:
\[
\lim_{n \to \infty} \Tr(\Phi^n\circ \hat{X}) = \Tr(\Phi_*\circ \hat{X}) = \sum_{\alpha, \beta =1}^{D}e_{\alpha}^{\dagger} \Phi_*\big( \hat{X}(e_{\alpha}e_{\beta}^{\dagger})\big)e_{\beta}.
\]
Utilizing the form of $\Phi_*$ from \eqref{eq_limMn}, we obtain:
\[
\lim_{n \to \infty} \Tr(\Phi^n\circ \hat{X}) = \sum_{\alpha, \beta =1}^{D} \Tr\big( \hat{X}(e_{\alpha}e_{\beta}^{\dagger})\big)e_{\alpha}^{\dagger}\rho_{*}e_{\beta}
\]
For the special case $X = \id_{[1,N]}$, where $\hat{X} = \Phi^N$, the normalizing constant satisfies:
\[
\mathcal{N}(N+n) = \Tr(\Phi^{N+n}) \longrightarrow \Tr(\Phi_{*}) = 1
\]
Consequently, the expectation value of $X$ in the evolved state converges as:
\[
\lim_{n\to\infty}\varphi_{n+N}(X) = \lim_{n\to\infty} \frac{\Tr(\Phi^{n}\circ\hat{X})}{\mathcal{N}(N+n)} = \sum_{\alpha, \beta =1}^{D}e_{\alpha}^{\dagger}\Phi_{*}\big(\hat{X}(e_{\alpha}e_{\beta}^{\dagger})\big)e_{\beta}
\]
Thus, the sequence of states $\{\varphi_n\}$ converges to the steady state $\varphi$, demonstrating thermalization of the system under repeated interactions.
\end{proof}

This result establishes a fundamental property of quantum dynamical systems: the tendency of a quantum channel to drive the system towards a unique equilibrium state. The proof highlights how ergodicity and mixing play a crucial role in ensuring the eventual stabilization of MPS, regardless of initial conditions.
Furthermore, the structure of the limiting dynamics, encoded in the stationary channel $\Phi_*$, offers insights into the underlying symmetries and invariant subspaces governing the system's behavior.

\section{Depolarizing Matrix Product States}\label{sect_depoMPS}

In quantum information theory, the depolarizing channel serves as a fundamental model for quantum noise, representing a process where a quantum state is replaced by a maximally mixed state with a certain probability. This section delves into the construction of Matrix Product States (MPS) associated with the depolarizing channel, elucidating their structure and significance.

Consider the Hilbert space $\mathbb{C}^2$ with its standard basis vectors:
\[
e_1 = \begin{pmatrix} 1 \\ 0 \end{pmatrix}, \quad e_2 = \begin{pmatrix} 0 \\ 1 \end{pmatrix}.
\]

The Pauli matrices, which form a basis for the space of $2 \times 2$ Hermitian matrices, are defined as:
\[
I = \begin{pmatrix} 1 & 0 \\ 0 & 1 \end{pmatrix}, \quad \sigma_x = \begin{pmatrix} 0 & 1 \\ 1 & 0 \end{pmatrix}, \quad \sigma_y = \begin{pmatrix} 0 & -i \\ i & 0 \end{pmatrix}, \quad \sigma_z = \begin{pmatrix} 1 & 0 \\ 0 & -1 \end{pmatrix}.
\]

The depolarizing channel $\Phi$ is a quantum channel that transforms a density matrix $\rho$ into a mixture of itself and the maximally mixed state, effectively modeling the loss of information in a quantum system. For a parameter $p \in [0,1]$, the action of $\Phi$ on a matrix $M \in \mathbb{M}_2(\mathbb{C})$ is given by:
\[
\Phi(M) = (1 - p) M + \frac{p}{3} (\sigma_x M \sigma_x^\dagger + \sigma_y M \sigma_y^\dagger + \sigma_z M \sigma_z^\dagger).
\]

This can be expressed using Kraus operators $\{A_k\}_{k=0}^3$ as:
\[
\Phi(M) = \sum_{k=0}^{3} A_k M A_k^\dagger,
\]
where
\[
A_0 = \sqrt{1 - p}\, I, \quad A_1 = \sqrt{\frac{p}{3}}\, \sigma_x, \quad A_2 = \sqrt{\frac{p}{3}}\, \sigma_y, \quad A_3 = \sqrt{\frac{p}{3}}\, \sigma_z.
\]

These operators satisfy the completeness relation:
\[
\sum_{k=0}^{3} A_k^\dagger A_k = I
\]
ensuring that $\Phi$ is trace-preserving. The $N$-site MPS associated with these tensors is constructed as:
\[
|\psi_N\rangle = \sum_{i_1, i_2, \ldots, i_N=0}^{3} \operatorname{Tr}(A_{i_1} A_{i_2} \cdots A_{i_N}) |i_1, i_2, \ldots, i_N\rangle,
\]
where $\{|i\rangle\}_{0 \leq i \leq 3}$ denotes the standard basis of the  4-dimensional Hilbert space $\mathcal{H} \equiv \mathbb{C}^4$.

This MPS encapsulates the effects of the depolarizing channel across multiple sites, providing a framework to analyze how local noise influences the global properties of quantum states. The structure of the MPS reflects the interplay between the unitary evolution represented by the Pauli matrices and the probabilistic noise parameterized by $p$.

The space of observables is defined by the algebra $\mathcal{B} = \mathcal{B}(\mathbb{C}^4)$, and its corresponding infinite tensor product structure is given by:
\[
\mathcal{B}_{\mathbb{N}} = \bigotimes_{n \in \mathbb{N}} \mathcal{B}(\mathbb{C}^4).
\]

The sequence of states associated with the depolarizing Matrix Product State (MPS), as formulated in (\ref{eq_varphi}), is expressed as:
\[
\varphi_n(X) = \frac{\big\langle \Psi_{n+1} \big| X \otimes \id_{n+1} \big| \Psi_{n+1} \big\rangle}{\mathcal{N}(n+1)},
\]
where $X$ is a local observable acting on the first $n$ sites, and $\mathcal{N}(n+1)$ is the normalization factor ensuring proper state definition.

\begin{theorem}
The sequence of states $\{\varphi_n\}$ associated with the depolarizing MPS converges in the weak-* topology to a limiting state $\varphi$ on $\mathcal{B}_{\mathbb{N}}$. This limiting state, when evaluated on a local observable $X$, is given by:
\begin{equation}\label{eq_depMPS}
\varphi(X) = \frac{1}{2} \sum_{\substack{i_1,\dots,i_{N} \\ j_1,\dots,j_{N}}} \big\langle i_1,\ldots, i_N \big| X \big| j_1,\ldots, j_N \big\rangle \operatorname{Tr}\big(A_{i_N}^{[N]\, \dagger} \cdots A_{i_1}^{[1]\, \dagger} A_{j_1}^{[1]} \cdots A_{j_N}^{[N]}\big),
\end{equation}
for all $X \in \mathcal{B}_{[1,N]}$.

This result establishes that the depolarizing MPS defines a well-posed infinite-volume quantum state, capturing its asymptotic structure through the behavior of its local restrictions.
\end{theorem}

\begin{proof}
The quantum channel $\Phi$ satisfies both the completely positive trace-preserving (CPTP) property and the completely positive identity-preserving (CPIP) condition, making it a bistochastic channel. For a given density matrix:
\[
\rho = \begin{pmatrix} \rho_{11} & \rho_{12} \\ \rho_{21} & \rho_{22} \end{pmatrix},
\]
the action of the quantum channel $\Phi$ on $\rho$ is expressed as:
\[
\Phi(\rho) = \begin{pmatrix} (1- \frac{2}{3}p)\rho_{11} + \frac{2}{3}p \rho_{22} & (1- \frac{4}{3}p)\rho_{12} \\
(1-\frac{4}{3}p)\rho_{21} & (1- \frac{2}{3}p)\rho_{22} + \frac{2}{3}p \rho_{11} \end{pmatrix}.
\]

The associated quantum Markov-Dobrushin constant is given by:
\[
\kappa_{\Phi} = \inf\left\{\Phi( \xi\xi^{\dagger}) : \xi \in \mathbb{C}^2, \|\xi\| = 1 \right\} = \frac{2p}{3} I.
\]

Two distinct cases arise:
\begin{itemize}
\item If $p \neq \frac{3}{4}$, the quantum channel $\Phi$ is mixing. For any initial state $\rho \in \mathfrak{S}_2$, the iterates $\Phi^n(\rho)$ converge to the unique fixed point
\(
\rho_* = \frac{1}{2} I.
\)

\item If $p = \frac{3}{4}$, then for every $\rho \in \mathfrak{S}_2$, we have:
\(\Phi(\rho) = \frac{1}{2} I\). This implies that $\Phi$ is stationary. In this case, the qubit channel coincides with the completely depolarizing channel, which maps any state to the maximally mixed state.
\end{itemize}
The limiting quantum channel of $\Phi^n$ is then  the completely depolarizing channel defined as:
\begin{equation}\label{eq-Omega}
\Omega(M) = \frac{1}{2} \operatorname{Tr}(M) I
\end{equation}
According to Theorem \ref{thm_ergMPS}, the limiting matrix product state on the quantum spin system $\mathcal{B}_{\mathbb{N}}$, as given by equation (\ref{eq_limMn}), satisfies:
\[
\varphi(X) = \frac{1}{2} \sum_{\alpha,\beta =1}^{D} e_{\alpha}^{\dagger} \rho_{*} e_{\beta} \operatorname{Tr}\big(\hat{X}( e_{\alpha} e_{\beta}^{\dagger})\big) = \frac{1}{2}\operatorname{Tr}\big(\hat{X}(I)\big),
\]
for every local observable $X \in \bigotimes_{[1,N]} \mathcal{B}(\mathbb{C}^4)$. From equation (\ref{eq-hatX}), one obtains (\ref{eq_depMPS}).
\end{proof}

This formulation provides a rigorous characterization of the limiting matrix product state under the depolarizing channel, emphasizing its mixing properties and its role in quantum spin systems.

\subsection{Non-Projectivity of the State Sequence \texorpdfstring{$\{\varphi_n\}_n$}{\{varphi\_n\}\_n}}

In the context of matrix product states (MPS), it is crucial to examine whether the sequence of states $\{\varphi_n\}_n$ exhibits the projective property. A sequence is termed projective if the limiting state $\varphi$, when restricted to the algebra $\mathcal{B}_{[1,n]}$, coincides with the finite state $\varphi_n$ for all $n$. However, our analysis reveals that $\{\varphi_n\}_n$ does not possess this property, even in the simplest case where $n=1$.

To illustrate this, consider the observable $X = |0\rangle\langle0| \in \mathcal{B}_1$. We begin by computing the normalization factor $\mathcal{N}(2)$ for $N=2$:
\[
\mathcal{N}(2) = \sum_{i_1, i_2=0}^{3} \left| \operatorname{Tr}(A_{i_1} A_{i_2}) \right|^2 = 4 \left( (1 - p)^2 + \frac{p^2}{3} \right),
\]
where $A_{i}$ are the Kraus operators associated with the depolarizing channel, and $p$ is the depolarization parameter.

Next, we evaluate the expectation value $\langle \psi_{2} | X \otimes \mathbb{I}_2 | \psi_{2} \rangle$:
\[
\langle \psi_{2} | X \otimes \mathbb{I}_2 | \psi_{2} \rangle = \sum_{i_1, j_1, j} \operatorname{Tr}(A_{i_1} A_{j}) \overline{\operatorname{Tr}(A_{j_1} A_{j})} \, \, \langle i_1 | X | j_1 \rangle = 4 (1 - p)^2
\]
The state $\varphi_1$ applied to $X$ is then:
\[
\varphi_{1}(X) = \frac{\langle \psi_{2} | X \otimes \mathbb{I}_2 | \psi_{2} \rangle}{\mathcal{N}(2)} = \frac{(1 - p)^2}{(1 - p)^2 + \frac{p^2}{3}}
\]
On the other hand, the limiting state $\varphi$ applied to $X$ yields:
\[
\varphi(X) = \frac{1}{2} \sum_{i_1, j_1} \langle i_1 | X | j_1 \rangle \operatorname{Tr}(A_{i_1}^{\dagger} A_{j_1}) = 1 - p.
\]

Clearly, $\varphi(X) \neq \varphi_1(X)$, demonstrating that the sequence $\{\varphi_n\}_n$ is not projective. This non-projectivity indicates that, unlike quantum Markov chains (QMCs) and finitely correlated states (FCS), matrix product states on quantum spin chains do not form a subclass of projective states. This distinction has significant implications for the analysis and simulation of quantum spin systems, as it suggests that MPS can capture a broader range of correlations and entanglement structures that projective states cannot encompass.

\end{document}